\documentclass[11pt,letterpaper]{article}
\usepackage[affil-it]{authblk}

\usepackage{chihao}

\usepackage[margin=1in]{geometry}

\usepackage{graphicx}
\usepackage{threeparttable}
\usepackage{multirow}
\usepackage{ragged2e}
\usepackage{array}
\usepackage{makecell}
\usepackage{physics}
\usepackage{xifthen}
\usepackage{booktabs}
\usepackage{makecell}
\usepackage{tikz}
\usepackage{rotating}
\usetikzlibrary{math}
\usetikzlibrary{positioning}
\usepackage[labelsep=period]{caption}

\usepackage{bbm}
\usepackage{subfig}
\usepackage{afterpage}
\usepackage{pifont}

\usepackage{enumerate}
\usepackage{comment}

\newcommand{\Bern}{\!{Bern}\xspace}
\newcommand{\Unif}{\!{Unif}\xspace}

\newcommand{\GainsFromTrade}{\textsf{Gains from Trade}\xspace}
\newcommand{\TotalGainsFromTrade}{\textsf{Total Gains from Trade}\xspace}
\newcommand{\GFT}{\textnormal{\textsf{GFT}}\xspace}
\newcommand{\SocialWelfare}{\textsf{Social Welfare}\xspace}
\newcommand{\SW}{\textnormal{\textsf{SW}}}
\newcommand{\Profit}{\textnormal{\textsf{Profit}}}
\newcommand{\Regret}{\textnormal{\textsf{Regret}}}

\newcommand{\Mech}{\+M}

\newcommand{\StrongBudgetBalance}{\textsf{Strong Budget Balance}\xspace}
\newcommand{\SBB}{\textsf{SBB}\xspace}
\newcommand{\WeakBudgetBalance}{\textsf{Weak Budget Balance}}
\newcommand{\WBB}{\textsf{WBB}\xspace}
\newcommand{\LocalBudgetBalance}{\textsf{Local Budget Balance}}
\newcommand{\LBB}{\textsf{LBB}\xspace}
\newcommand{\GlobalBudgetBalance}{\textsf{Global Budget Balance}\xspace}
\newcommand{\GBB}{\textsf{GBB}\xspace}

\newcommand{\Val}[1]{(S^{#1}, B^{#1})}
\newcommand{\SVal}[1]{S^{#1}}
\newcommand{\BVal}[1]{B^{#1}}

\newcommand{\Price}[1]{(P^{#1}, Q^{#1})}
\newcommand{\SPrice}[1]{P^{#1}}
\newcommand{\BPrice}[1]{Q^{#1}}

\newcommand{\Feedback}[1]{(X^{#1}, Y^{#1})}
\newcommand{\SFeedback}[1]{X^{#1}}
\newcommand{\BFeedback}[1]{Y^{#1}}

\newcommand{\Trade}[2][]{Z_{#1}^{#2}}

\newcommand{\GBBSemi}{\textnormal{\textsc{GBB-Semi}}}
\newcommand{\ProfitMax}{\textnormal{\textsc{ProfitMax}}}

\renewcommand{\Tilde}{\widetilde}

\renewcommand{\Hat}{\widehat}

\newcommand{\ignore}[1]{}

\newcommand{\tO}{\Tilde{\+O}}

\newcommand{\tTheta}{\Tilde{\Theta}}

\algrenewcommand\algorithmicrequire{\textbf{Input:}}
\algrenewcommand\algorithmicensure{\textbf{Output:}}

\makeatletter
\def\term{\@ifnextchar[\term@optarg\term@noarg}
\def\term@optarg[#1]#2{%
  \textup{#1}%
  \def\@currentlabel{#1}%
  \def\cref@currentlabel{[][2147483647][]#1}%
  \cref@label[term]{#2}}
\def\term@noarg#1{%
  \refstepcounter{termcounter}%
  \textup{\thetermcounter}%
  \cref@label[term]{#1}}
\makeatother

\def\colorful{1}
\ifnum\colorful=1

\fi
\ifnum\colorful=0

\fi

\usepackage{pict2e}

\makeatletter
\DeclareRobustCommand{\Circle}{%
  \mathbin{\mathpalette\on@ntimes\relax}%
}
\newcommand{\on@ntimes}[2]{%
  \vcenter{\hbox{%
    \sbox0{\m@th$#1\otimes$}%
    \setlength\unitlength{\wd0}%
    \begin{picture}(1,1)
    \linethickness{0.5pt}
    \put(.5,.5){\circle{.8}}
    \end{picture}%
  }}%
}
\makeatother

\def\version{Full}

\title{Tight Regret Bounds for Bilateral Trade under Semi Feedback}

\ifthenelse{\equal{\version}{Full}}
{\author{
Yaonan Jin\thanks{Hong Kong University of Science and Technology. Email: {\tt jinyaonan1996@gmail.com}}
}}
{\ifthenelse{\equal{\version}{Anonymous}}
{\author{
Anonymous Submission
}}
{}}

\date{}

\begin{document}
	
\maketitle
\begin{abstract}
The study of \textit{regret minimization in fixed-price bilateral trade} has received considerable attention in recent research. Previous works \cite{CCCFL24mor, CCCFL24jmlr, AFF24, BCCF24, CJLZ25, LCM25b, DDFS25} have acquired a thorough understanding of the problem, except for determining the tight regret bound for {\GBB} semi-feedback fixed-price mechanisms under adversarial values.

In this paper, we resolve this open question by devising an $\tO(T^{2 / 3})$-regret mechanism, matching the $\Omega(T^{2 / 3})$ lower bound from \cite{CJLZ25} up to polylogarithmic factors.
\end{abstract}


\section{Introduction}
\label{sec:intro}

Repeated bilateral trade is a canonical model in mechanism design. In each round $t \in [T]$, a (new) seller and a (new) buyer seek to trade an indivisible item; the seller values the item at $\SVal{t}$, and the buyer values it at $\BVal{t}$. The goal is to improve \textit{economic efficiency} through a mechanism, measured by two standard metrics:
\begin{itemize}
    \item {\GainsFromTrade}: $\GFT = \sum_{t \in [T]} \GFT^{t} = \sum_{t \in [T]} (\BVal{t} - \SVal{t}) \cdot \Trade[]{t}$.

    \item {\SocialWelfare}: $\SW = \sum_{t \in [T]} \SW^{t} = \sum_{t \in [T]} (\BVal{t} \cdot \Trade[]{t} + \SVal{t} \cdot (1 - \Trade[]{t}))$.
\end{itemize}
Here, $\Trade[]{t} \in \{0, 1\}$ indicates whether a trade succeeds. The \textit{(additive) regret minimization} framework unifies both metrics since the differences $\SW^{t} - \GFT^{t} = \SVal{t}$ are independent of the mechanism. Without loss of generality, we adopt {\GainsFromTrade} for our presentation.

The values $\Val{t}_{t \in [T]}$ are conventionally modeled in one of three settings, listed below from most to least general ones (see \Cref{sec:prelim} for details):
\begin{align*}
    \text{``adversarial values''}
    \quad\subseteq\quad
    \text{``correlated values''}
    \quad\subseteq\quad
    \text{``independent values''}.
\end{align*}

A large body of literature focuses on the family of \textit{fixed-price mechanisms},\footnote{This focus is largely due to their ``economic viability'', satisfying both \textit{Ex-Post Individual Rationality} and \textit{Dominant-Strategy Incentive Compatibility}.} and so does this paper. For brevity, we use \textit{mechanism} to mean \textit{fixed-price mechanism} when clear from context.
In each round $t \in [T]$, the mechanism posts two prices: $\SPrice{t}$ for the seller and $\BPrice{t}$ for the buyer. A trade succeeds if both accept: $\Trade[]{t} = {\bb 1}[\SVal{t} \le \SPrice{t} \land \BPrice{t} \le \BVal{t}]$, yielding $\GFT^{t} = (\BVal{t} - \SVal{t}) \cdot \Trade[]{t}$.

An ``economically viable'' mechanism must fulfill the \textit{Budget Balance} constraint. Prior works \cite{MS83, HR87, CKLT16, CCCFL24mor, CCCFL24jmlr, AFF24, BCCF24, CJLZ25, LCM25a, LCM25b} have studied two main notions---the former is stricter than the latter (see \Cref{sec:prelim} for details):
\begin{align*}
    \text{{\LocalBudgetBalance} ({\LBB})}
    \quad\subseteq\quad
    \text{{\GlobalBudgetBalance} ({\GBB})}.
\end{align*}

Mechanism design and analysis further depend on the underlying \textit{feedback model}---the information revealed after each round. Three feedback models have been widely studied \cite{BSZ06, CCCFL24mor, CCCFL24jmlr, AFF24, BCCF24, CJLZ25, LCM25a, LCM25b, DDFS25}, listed below from most to least informative ones (see \Cref{sec:prelim} for details):\footnote{The semi feedback model was introduced independently in \cite{CJLZ25} (where the name originated) and \cite{LCM25a} (under the name \textit{asymmetric feedback}).}
\begin{align*}
    \text{full feedback}
    \quad\succeq\quad
    \text{semi feedback}
    \quad\succeq\quad
    \text{partial feedback}.
\end{align*}

For full and partial feedback, a series of works \cite{BSZ06, CCCFL24mor, CCCFL24jmlr, AFF24, BCCF24, CJLZ25, LCM25b, DDFS25} has provided an essentially complete picture of regret bounds (see \Cref{tbl:full,tbl:partial} for an overview).
However, the understanding of the intermediate \textit{semi-feedback} model remains incomplete.
Notably, \cite{CJLZ25, LCM25a} explicitly posed the following open question:
\begin{quote}
    \textbf{\textit{What is the tight regret for {\GBB} semi-feedback mechanisms under adversarial values?}}
\end{quote}

In this paper, we answer this question (up to polylogarithmic factors) by devising an $\tO(T^{2 / 3})$-regret mechanism. Technically, our mechanism modifies the canonical Exp3 algorithm \cite[Chapter~11]{LS20} in a nontrivial way. Combined with existing progress \cite{CJLZ25,LCM25a},\footnote{In more detail, although \cite{LCM25a} investigated a slightly different context, Algorithm~3 of that work readily translates to an $\tO(T^{2 / 3})$-regret {\LBB} (more precisely, {\StrongBudgetBalance} (\Cref{sec:prelim})) semi-feedback mechanism under correlated values.} our result completes the landscape of regret bounds for semi-feedback mechanisms, as summarized in \Cref{tbl:semi}.

\begin{table}[t]
    \centering
    \begin{tabular}{|c||>{\centering}p{5.5cm}|>{\centering\arraybackslash}p{5.5cm}|}
        \hline
        \rule{0pt}{12pt} & {\LocalBudgetBalance} & {\GlobalBudgetBalance} \\ [1pt]
        \hline
        \hline
        \rule{0pt}{12pt}Independent & \multicolumn{2}{c|}{\multirow{2}*{\rule{0pt}{16pt}$T^{1 / 2}$ \cite{CCCFL24mor, BCCF24, CJLZ25}}} \\ [1pt]
        \cline{1-1}
        \rule{0pt}{12pt}Correlated & \multicolumn{2}{c|}{} \\ [1pt]
        \cline{1-2}
        \rule{0pt}{12pt}Adversarial & $T$ \cite{CCCFL24mor, AFF24} &  \\ [1pt]
        \hline
    \end{tabular}
    \caption{\label{tbl:full}Nearly tight regret bounds of \textit{full-feedback} fixed-price mechanisms.}
    \vspace{.1in}
    \begin{tabular}{|c||>{\centering\arraybackslash}p{5.5cm}|>{\centering\arraybackslash}p{5.5cm}|}
        \hline
        \rule{0pt}{12pt} & {\LocalBudgetBalance} & {\GlobalBudgetBalance} \\ [1pt]
        \hline
        \hline
        \rule{0pt}{12pt}Independent &  & $T^{2 / 3}$ \cite{CJLZ25} \\ [1pt]
        \cline{1-1}\cline{3-3}
        \rule{0pt}{12pt}Correlated & $T$ \cite{CCCFL24mor, AFF24, CJLZ25} & \multirow{2}*{\rule{0pt}{16pt}$T^{3 / 4}$ \cite{BCCF24, CJLZ25, LCM25b}} \\ [1pt]
        \cline{1-1}
        \rule{0pt}{12pt}Adversarial &  & \\ [1pt]
        \hline
    \end{tabular}
    \caption{\label{tbl:partial}Nearly tight regret bounds of \textit{partial-feedback} fixed-price mechanisms.}
    \vspace{.1in}
    \begin{tabular}{|c||>{\centering}p{5.5cm}|>{\centering\arraybackslash}p{5.5cm}|}
        \hline
        \rule{0pt}{12pt} & {\LocalBudgetBalance} & {\GlobalBudgetBalance} \\ [1pt]
        \hline
        \hline
        \rule{0pt}{12pt}Independent & \multicolumn{2}{c|}{\multirow{2}*{\rule{0pt}{16pt}$T^{2 / 3}$ \cite{CJLZ25, LCM25a}}} \\ [1pt]
        \cline{1-1}
        \rule{0pt}{12pt}Correlated & \multicolumn{2}{c|}{} \\ [1pt]
        \hline
        \rule{0pt}{12pt}Adversarial & $T$ \cite{CCCFL24mor, AFF24} & $T^{2 / 3}$ [\Cref{thm:GBB-adversarial}] \\ [1pt]
        \hline
    \end{tabular}
    \caption{\label{tbl:semi}Nearly tight regret bounds of \textit{semi-feedback} fixed-price mechanisms.}
\end{table}

\begin{remark}[Motivation for Semi Feedback]
The semi feedback model, as a natural informational intermediate, is directly motivated by prevalent economic scenarios.

A primary motivation occurs when the mechanism designer is itself a trading party---acting either as the seller (e.g., a digital storefront) or as the buyer (e.g., in procurement). In such cases, the designer knows its own valuation and observes only whether the counterparty accepts the trade, corresponding precisely to semi feedback $(\SVal{t}, \Trade{t})$ or $(\Trade{t}, \BVal{t})$.

Even for a neutral platform, semi feedback emerges under common institutional asymmetries. In regulated financial markets (e.g., block trading), sellers often disclose exact reservation prices while buyers' bids remain private beyond the trade outcome---yielding $(\SVal{t}, \Trade{t})$. In government procurement, the agency's budget is public while suppliers' costs are only revealed through participation---corresponding to $(\Trade{t}, \BVal{t})$. Similarly, an online marketplace typically knows sellers' listed prices and trade outcomes, but rarely observes buyers' true valuations.

These ubiquitous settings---where one side's value is transparent while the other's is inferred only from a binary outcome---establish semi feedback as a realistic and distinct informational model, justifying its independent study between the extremes of full and partial feedback.
\end{remark}

\section{Notations and Preliminaries}
\label{sec:prelim}

For two nonnegative integers $m \ge n \ge 0$, define the sets $[n \colon m] \defeq \{n, n + 1, \dots, m - 1, m\}$ and $[n] \defeq [1 \colon n] = \{1, 2, \dots, n\}$.
Given a (possibly random) event $\+E$, let ${\bb 1}[\+E] \in \{0, 1\}$ denote the indicator function. For a real number $x \in {\bb R}$, let $[x]_{+} \defeq \max\{x, 0\}$.

\vspace{.1in}
\noindent
{\bf Repeated Bilateral Trade.}
We study a $T$-round repeated game against an oblivious adversary; throughout, we assume that $T \gg 1$ is a sufficiently large integer.
In each round $t \in [T]$, a (new) seller and a (new) buyer seek to trade an indivisible item. The seller values the item at $\SVal{t}$, and the buyer values it at $\BVal{t}$.
Our goal is to maximize \textit{economic efficiency} by designing a mechanism that ideally trades whenever $\SVal{t} \le \BVal{t}$.
The adversary controls the generation of the values $\Val{t}_{t \in [T]}$ according to one of three classic models, listed below from most to least general ones:
\begin{flushleft}
\begin{itemize}
    \item \textbf{Adversarial Values:}
    The adversary chooses an arbitrary joint distribution $\+D$ supported on $[0, 1]^{2T}$. The values $\Val{t}_{t \in [T]}$ are then drawn from $\+D$.\footnote{This definition emphasizes the generality of the ``adversarial values'' model over the ``correlated/independent values'' models. From the perspective of (additive) regret minimization, however, the worst-case distribution $\+D_{*}$ can be assumed, without loss of generality, to place all its mass on a finite set of $2T$ discrete points $(\SVal{t}_{*}, \BVal{t}_{*}) \in [0, 1]^{2T}$.}
    
    \item \textbf{Correlated Values:}
    The adversary chooses an arbitrary joint distribution $\+D$ supported on $[0, 1]^{2}$. The values $\Val{t}$ in each round $t \in [T]$ are drawn independently and identically from $\+D$.
    
    \item \textbf{Independent Values:}
    This is identical to the ``correlated values'' setting, except that $\+D$ must be a product distribution $\+D_{S} \bigotimes \+D_{B}$, making all $2T$ values $\Val{t}_{t \in [T]}$ mutually independent.
\end{itemize}
\end{flushleft}

\noindent
{\bf Fixed-Price Mechanisms.}
A fixed-price mechanism $\Mech$ posts two (possibly randomized) prices in each round $t \in [T]$: $\SPrice{t}$ for the seller and $\BPrice{t}$ for the buyer. A trade occurs if both agents accept their respective prices.
For any prices $\Price{t} \in [0, 1]^{2}$ in a single round, this yields the \textit{\GainsFromTrade} $\GFT^{t}\Price{t} \in [0, 1]$ and the \textit{profit} $\Profit^{t}\Price{t} \in [-1, 1]$:
\begin{align*}
    \textstyle
    \GFT^{t}\Price{t}
    &\textstyle ~\defeq~ (\BVal{t} - \SVal{t}) \cdot {\bb 1}[\SVal{t} \le \SPrice{t}] \cdot {\bb 1}[\BPrice{t} \le \BVal{t}], \\
    \textstyle
    \Profit^{t}\Price{t}
    &\textstyle ~\defeq~ (\BPrice{t} - \SPrice{t}) \cdot {\bb 1}[\SVal{t} \le \SPrice{t}] \cdot {\bb 1}[\BPrice{t} \le \BVal{t}].
\end{align*}
Initially, the mechanism does not know the underlying distribution $\+D$ nor the realized values $\Val{t}_{t \in [T]}$. After each round $t \in [T]$, it receives some feedback based on the prices $\Price{t}$ and, then, proceeds to round $t + 1$ (and computes new prices $\Price{t + 1}$).
We study three natural and widely adopted feedback models, listed below from most to least informative ones.
\begin{flushleft}
\begin{itemize}
    \item \textbf{Full Feedback:}
     $\Val{t} \in [0, 1]^{2}$ reveals both agents' values.
    
    \item \textbf{Semi Feedback:}
    This model has four specific types:
    \begin{itemize}
        \item $(\SVal{t}, \BFeedback{t}) \in [0, 1] \times \{0, 1\}$ reveals the seller's value $\SVal{t}$ and the buyer's trading intention $\BFeedback{t} = \BFeedback{t}(\BPrice{t}) \defeq {\bb 1}[\BPrice{t} \le \BVal{t}]$.
        
        \item $(\SFeedback{t}, \BVal{t}) \in \{0, 1\} \times [0, 1]$ reveals the seller's trading intention $\SFeedback{t} = \SFeedback{t}(\SPrice{t}) \defeq {\bb 1}[\SVal{t} \le \SPrice{t}]$ and the buyer's value $\BVal{t}$.
        
        \item $(\SVal{t}, \Trade{t}) \in [0, 1] \times \{0, 1\}$ reveals the seller's value $\SVal{t}$ and the trade outcome $\Trade[]{t} = \Trade[]{t}(\SPrice{t}, \BPrice{t}) \defeq {\bb 1}[\SVal{t} \le \SPrice{t} \land \BPrice{t} \le \BVal{t}] \equiv \SFeedback{t} \land \BFeedback{t}$.
        
        \item $(\Trade{t}, \BVal{t}) \in \{0, 1\} \times [0, 1]$ reveals the trade outcome $\Trade[]{t} \equiv \SFeedback{t} \land \BFeedback{t}$ and the buyer's value $\BVal{t}$.
    \end{itemize}
    
    \item \textbf{Partial Feedback:}
    This model has two specific types:
    \begin{itemize}
        \item \textbf{Two-Bit Feedback:}
        $\Feedback{t} \in \{0, 1\}^{2}$ reveals both agents' trading intentions.
        
        \item \textbf{One-Bit Feedback:}
        $\Trade[]{t} \equiv \SFeedback{t} \land \BFeedback{t} \in \{0, 1\}$ reveals the trade outcome.
    \end{itemize}
\end{itemize}
\end{flushleft}
Up to polylogarithmic factors, two-bit feedback and one-bit feedback yield identical regret bounds in all settings considered (\Cref{tbl:full,tbl:partial,tbl:semi}), so we unify them under the term \textit{partial feedback} to reflect this regret-equivalence. Similarly for the four specific types of \textit{semi feedback}.

\begin{figure}[t]
{\centering
\begin{tikzpicture}[scale = 0.8]
    \node[right] at (-10.5, 0) {full feedback};
    \node[right] at (-10.5, -2.25) {semi feedback};
    \node[right] at (-10.5, -5.25) {partial feedback};
    \draw[densely dotted] (-10.5, -0.75) to (7, -0.75);
    \draw[densely dotted] (-10.5, -3.75) to (7, -3.75);
    
    \node(n1) at (0, 0) {$\Val{t}$};
    \node(n21) at (-3, -1.5) {$(\SVal{t}, \BFeedback{t})$};
    \node(n22) at (3, -1.5) {$(\SFeedback{t}, \BVal{t})$};
    \node(n23) at (-6, -3) {$(\SVal{t}, \Trade[]{t})$};
    \node(n24) at (6, -3) {$(\Trade[]{t}, \BVal{t})$};
    \node(n31) at (0, -4.5) {$\Feedback{t}$};
    \node(n32) at (0, -6) {$\Trade[]{t}$};
    
    \draw[-latex] (n1.south west) to (n21.north east);
    \draw[-latex] (n1.south east) to (n22.north west);
    \draw[-latex] (n21.south west) to (n23.north east);
    \draw[-latex] (n21.south east) to (n31.north west);
    \draw[-latex] (n22.south west) to (n31.north east);
    \draw[-latex] (n22.south east) to (n24.north west);
    \draw[-latex] (n23.south east) to (n32.north west);
    \draw[-latex] (n31.south) to (n32.north);
    \draw[-latex] (n24.south west) to (n32.north east);
\end{tikzpicture}
\par}
\caption{A Hasse diagram of feedback models. An arrow $\+F \to \+F'$ indicates that $\+F$ is more informative than $\+F'$; for example, $\Val{t} \to (\SVal{t}, \BFeedback{t})$ since $\BFeedback{t} = {\bb 1}[\BPrice{t} \le \BVal{t}]$ is determined by $\BVal{t}$ and $\BPrice{t}$.}
\label{fig:feedback}
\end{figure}

For economic viability, we impose the \textit{Budget Balance} constraint. We consider two main notions: the stricter \textit{{\LocalBudgetBalance} ({\LBB})} \cite{MS83} and the looser \textit{{\GlobalBudgetBalance} ({\GBB})} \cite{BCCF24}.
\begin{flushleft}
\begin{itemize}
    \item {\LocalBudgetBalance}:
    This can be subdivided into the stricter \textit{{\StrongBudgetBalance} ({\SBB})} and the looser \textit{{\WeakBudgetBalance} ({\WBB})} constraint.
    \begin{itemize}
        \item {\StrongBudgetBalance}: 
        $\Profit^{t}\Price{t} = 0$ (essentially, $\SPrice{t} = \BPrice{t}$) ex post, $\forall t \in [T]$.\\
        This requires zero profit $\Profit^{t}\Price{t} = 0$ in each round.
        
        \item {\WeakBudgetBalance}: 
        $\Profit^{t}\Price{t} \ge 0$ (essentially, $\SPrice{t} \le \BPrice{t}$) ex post, $\forall t \in [T]$.\\
        This requires nonnegative profit $\Profit^{t}\Price{t} \ge 0$ in each round.
    \end{itemize}
    
    \item {\GlobalBudgetBalance}:
    $\sum_{t \in [T]} \Profit^{t}\Price{t} \ge 0$ ex post.\\
    This requires nonnegative total profit $\sum_{t \in [T]} \Profit^{t}\Price{t} \ge 0$.
\end{itemize}
\end{flushleft}
Again, we unify {\SBB} and {\WBB} under the term {\LBB} to reflect their regret-equivalence (\Cref{tbl:full,tbl:partial,tbl:semi}).

\vspace{.1in}
\noindent
\textbf{Regret Minimization.}
We evaluate the economic efficiency of a fixed-price mechanism $\Mech$ within the regret minimization framework. Over the entire time horizon $[T]$, the mechanism induces \textit{\TotalGainsFromTrade} $\GFT_{\+D}^{\Mech}$ in expectation (over all possible randomness $\Val{t}_{t \in [T]} \sim \+D$ and $\Price{t}_{t \in [T]} \gets \Mech$).
We compare this to the \textit{\textsf{Bayesian-Optimal Total Gains from Trade}} $\GFT_{\+D}^{*}$ with knowledge of the distribution $\+D$.
\begin{align*}
    \textstyle
    \GFT_{\+D}^{\Mech}
    &\textstyle ~\defeq~ {\bb E}_{\Val{t}_{t \in [T]} \sim \+D,\; \Price{t}_{t \in [T]} \gets \Mech}\big[\sum_{t \in [T]}\GFT^{t}\Price{t}\big],\\
    \textstyle
    \GFT_{\+D}^{*}
    &\textstyle ~\defeq~ \max_{0 \le p, q \le 1} {\bb E}_{\Val{t}_{t \in [T]} \sim \+D}\big[\sum_{t \in [T]}\GFT^{t}(p, q)\big].
\end{align*}

\begin{definition}[Benchmarks]
\label{def:benchmark}
The benchmark $\GFT_{\+D}^{*}$ is robust to the {\SBB}/{\WBB}/{\GBB} constraints. (In sharp contrast, these constraints do impact the design and analysis of mechanisms.) Even the Bayesian-optimal {\GBB} fixed prices $(p^{*}, q^{*})$ can be made to satisfy the {\SBB} constraint $p^{*} = q^{*}$. To see that the benchmark $\GFT_{\+D}^{*}$ is invariant to these constraints, consider any {\GBB} fixed prices $(p, q) \in [0, 1]^{2}$:\\
(i)~If $p < q$, then the {\SBB} fixed prices $(p, p)$ induce at least the same {\TotalGainsFromTrade}.\\
(ii)~If $p > q$, any successful trade ($\Trade[]{t} = 1$) would violate the {\GBB} constraint because $\Profit^{t} = -(p - q) < 0$. If no trade ever occurs, both $(p, q)$ and the {\SBB} fixed prices $(p, p)$ yield zero {\TotalGainsFromTrade}.
\end{definition}

For a mechanism $\Mech$, we can define its \textit{(worst-case) regret} $\Regret^{\Mech}$ by considering  all possible distributions $\+D$. Our goal is to find the \textit{minimax regret} $\Regret^{*}$ by designing a regret-optimal mechanism:
\begin{align*}
    \textstyle
    \Regret^{\Mech}
    &\textstyle ~\defeq~ \max_{\+D} \big(\GFT_{\+D}^{*} - \GFT_{\+D}^{\Mech}\big),\\
    \textstyle
    \Regret^{*}
    &\textstyle ~\defeq~ \min_{\Mech} \Regret^{\Mech}.
\end{align*}

\section{\texorpdfstring{$\tO(T^{2 / 3})$}{} {\GBB} Semi-Feedback Upper Bound for Adversarial Values}
\label{sec:GBB-adversarial}

In this section, we examine the power of ``fixed-price mechanisms with the {\GlobalBudgetBalance} ({\GBB}) constraint and semi feedback'' in the ``adversarial values'' setting. Specifically, we will establish (\Cref{thm:GBB-adversarial}) the following algorithmic result.

\begin{theorem}[{\GBB} Semi-Feedback Upper Bound for Adversarial Values]
\label{thm:GBB-adversarial}
\begin{flushleft}
In the ``adversarial values'' setting, there is a ``{\GBB} semi-feedback fixed-price mechanism'' achieving $\tO(T^{2 / 3})$ regret.
\end{flushleft}
\end{theorem}

\noindent
This matches the $\Omega(T^{2 / 3})$ lower bound (even in the more restricted ``independent values'' setting) established in previous work \cite[Theorem~20]{CJLZ25}, up to polylogarithmic factors.

Among \textit{all the four types of semi feedback} (\Cref{sec:prelim}), without loss of generality, we consider $(\SVal{t}, \Trade{t})$,\footnote{For the other three types: $(\Trade{t}, \BVal{t})$ is symmetric, while $(\SVal{t}, \BFeedback{t})$ and $(\SFeedback{t}, \BVal{t})$ are more informative than $(\SVal{t}, \Trade{t})$ and $(\Trade{t}, \BVal{t})$, respectively---any upper bound for the latter immediately yields the same upper bound for the former.}
i.e., the seller's value $\SVal{t} \in [0, 1]$ and the trade outcome $\Trade[]{t} = \Trade[]{t}\Price{t} = {\bb 1}[\SVal{t} \le \SPrice{t} \land \BPrice{t} \le \BVal{t}] \in \{0, 1\}$.

Since the values $\Val{t}_{t \in [T]} \in [0, 1]^{2T}$ are adversarially given, we omit their randomness. In addition to the notations from \Cref{sec:prelim}, we further denote by $(p^{*}, q^{*}) \in [0, 1]^{2}$ the benchmark prices---note that $p^{*} = q^{*}$ (\Cref{def:benchmark})---and $\Regret^{t}\Price{t}$, $\forall t \in [T]$ the per-round regret.
\begin{align*}
    \textstyle
    (p^{*}, q^{*})
    &\textstyle ~\defeq~ \argmax_{0 \le p = q \le 1} \sum_{t \in [T]} \GFT^{t}(p, q), \\
    \textstyle
    \Regret^{t}\Price{t}
    &\textstyle ~\defeq~ \GFT^{t}(p^{*}, q^{*}) - \GFT^{t}\Price{t}.
\end{align*}

\subsection{Mechanism Design}

Our fixed-price mechanism, called {\GBBSemi} and shown in \Cref{alg:GBB-adversarial}, is built on the mechanism design framework proposed by \cite[Section~3]{BCCF24}.
 This fixed-price mechanism {\GBBSemi} has two phases:

\noindent
(Line~\ref{alg:GBB-adversarial:profit}) The first phase invokes a subroutine called {\ProfitMax}, whose performance guarantees are given in \Cref{prop:BCCF24}. {\ProfitMax} induces \textit{nonnegative} profit $\ge 0$ in each round (\Cref{prop:BCCF24:1} of \Cref{prop:BCCF24}), and its main purpose is to accumulate sufficient profit while just incurring tolerable regret.
This accumulated profit helps fulfill the {\GBB} constraint, making the design of the second phase more flexible.

\noindent
(Lines~\ref{alg:GBB-adversarial:loop} to \ref{alg:GBB-adversarial:estimation}) The second phase adapts the classic Exp3 algorithm \cite[Chapter~11]{LS20} in a nontrivial way. Remarkably, the profit accumulated in the first phase will not be exhausted (\Cref{lem:GFT-adversarial:GBB,rmk:GFT-adversarial:GBB}), thus fulfilling the {\GBB} constraint.

{\GBBSemi} uses the following parameters. Both phases share a \textit{discretization parameter} $K = \tTheta(T^{1 / 3})$. In addition, the first phase employs a \textit{profit threshold} $\beta = \tTheta(T^{2 / 3})$, while the second phase uses a \textit{learning rate} $\eta = \tTheta(T^{-2 / 3})$ and an \textit{exploration rate} $\gamma = \tTheta(T^{-1 / 3})$.
\begin{align*}
    K &\textstyle ~\defeq~ \frac{1}{4}T^{1 / 3} \log^{-2 / 3}(T),
    \tag{the discretization parameter} \\
    \beta &\textstyle ~\defeq~ \frac{3T}{K + 1},
    \tag{the profit threshold} \\
    \eta &\textstyle ~\defeq~ (\frac{\ln(K)}{T (K + 1)})^{1 / 2},
    \tag{the learning rate} \\
    \gamma &\textstyle ~\defeq~ \frac{1}{K + 1}.
    \tag{the exploration rate}
\end{align*}

To adapt the classic Exp3 algorithm \cite[Chapter~11]{LS20} in the second phase, a main obstacle is how to construct good estimators for the benchmark $\GFT^{t}(p^{*}, q^{*})$ using semi feedback $(\SVal{t}, \Trade{t})$ while incurring low regret.
Our strategy is to construct a $\frac{1}{K}$-net $\{(\frac{i}{K}, \frac{j - 1}{K})\}_{1 \le i, j \le K}$ of the whole action space $[0, 1]^{2}$; then, the benchmark action $(p^{*}, q^{*})$ is close to the discrete action $(\frac{k^{*}}{K}, \frac{k^{*} - 1}{K})$, where the index
\begin{align*}
    \textstyle
    k^{*} ~\defeq~ \max\{\lceil K p^{*} \rceil,\ 1\}
    ~=~ \max\{\lceil K q^{*} \rceil,\ 1\}.
\end{align*}
Conceivably, both actions' {\GainsFromTrade} $\GFT^{t}(p^{*}, q^{*})$ and $\GFT^{t}(\frac{k^{*}}{K}, \frac{k^{*} - 1}{K})$ are also close (cf.\ \Cref{lem:GFT-adversarial:discretization}).
Meanwhile, since this index $k^{*} \in [K]$ only has $K = \tTheta(T^{1 / 3})$ many possibilities, our strategy also simplifies the \textit{``exploration and exploitation''} of {\GBBSemi}, in alignment of the Exp3 algorithm.

As it turns out, under semi feedback $(\SVal{t}, \Trade{t})$, the {\GainsFromTrade}'s $\GFT^{t}(\frac{k}{K}, \frac{k - 1}{K})$ from such near-diagonal discrete actions $\{(\frac{k}{K}, \frac{k - 1}{K})\}_{k \in [K]}$ are still difficult to estimate unbiasedly; instead, the actual quantities to be estimated are the following variants:
\begin{align*}
    \textstyle
    \Tilde{\GFT}_{k}^{t}
    &\textstyle ~\defeq~ \underbrace{[\BVal{t} - \frac{k - 1}{K}]_{+} \cdot {\bb 1}[\SVal{t} \le \frac{k}{K}]}_{(\dagger)}
    ~+~ \underbrace{[\frac{k}{K} - \SVal{t}]_{+} \cdot {\bb 1}[\frac{k - 1}{K} \le \BVal{t}]}_{(\ddagger)},
    && \forall k \in [K].
\end{align*}
Below, \Cref{lem:GFT-adversarial:discretization} will justify that such variants $\Tilde{\GFT}_{k}^{t}$ are good surrogates for $\GFT^{t}(\frac{k}{K}, \frac{k - 1}{K})$, and \Cref{lem:GFT-adversarial:estimation} will show that the estimates $\Hat{\GFT}_{k}^{t}$ derived in Line~\ref{alg:GBB-adversarial:estimation} are \textit{unbiased estimators} for $\Tilde{\GFT}_{k}^{t}$, as desired.

Remarkably, as semi feedback $(\SVal{t}, \Trade[]{t})$ is asymmetric---the seller's value $\SVal{t}$ and the trade outcome $\Trade[]{t}$---we also estimate the two terms $(\dagger)$ and $(\ddagger)$ of $\Tilde{\GFT}_{k}^{t}$ in an asymmetric manner: \\
$(\dagger)$ is estimated by the first term in Line~\ref{alg:GBB-adversarial:estimation}, \textit{in each ``right-boundary'' round $t$ with $A^{t} = 0$} (Lines~\ref{alg:GBB-adversarial:case1} and \ref{alg:GBB-adversarial:exploration}). \\
$(\ddagger)$ is estimated by the second term in Line~\ref{alg:GBB-adversarial:estimation}, \textit{in each ``near-diagonal'' round $t$ with $A^{t} = 1$} (Lines~\ref{alg:GBB-adversarial:case0} and \ref{alg:GBB-adversarial:exploitation}).

\begin{lemma}[Discretization Errors]
\label{lem:GFT-adversarial:discretization}
\begin{flushleft}
In each round $t \in [T]$:
\begin{align*}
    \Tilde{\GFT}_{k}^{t}
    &\textstyle ~\le~ \GFT^{t}(\frac{k}{K}, \frac{k - 1}{K}) + \frac{1}{K}
    ~\le~ 1 + \frac{1}{K},
    && \forall k \in [K], \\
    \textstyle
    \Tilde{\GFT}_{k^{*}}^{t}
    &\textstyle ~\ge~ \GFT^{t}(p^{*}, q^{*}).
\end{align*}
\end{flushleft}
\end{lemma}

\begin{proof}
For the first equation (recall that $\GFT^{t}(\frac{k}{K}, \frac{k - 1}{K})$ is bounded between $[0, 1]$), we deduce that
\begin{align*}
    \textstyle
    \GFT^{t}(\frac{k}{K}, \frac{k - 1}{K})
    &\textstyle ~=~ ((\BVal{t} - \frac{k - 1}{K}) + (\frac{k}{K} - \SVal{t}) - \frac{1}{K}) \cdot {\bb 1}[\SVal{t} \le \frac{k}{K} \land \frac{k - 1}{K} \le \BVal{t}] \\
    \mr{definition of $\Tilde{\GFT}_{k}^{t}$}
    &\textstyle ~=~ \Tilde{\GFT}_{k}^{t} - \frac{1}{K} \cdot {\bb 1}[\SVal{t} \le \frac{k}{K} \land \frac{k - 1}{K} \le \BVal{t}] \\
    &\textstyle ~\ge~ \Tilde{\GFT}_{k}^{t} - \frac{1}{K}.
\end{align*}
The index $k^{*} \in [K]$ by definition satisfies $\frac{k^{*} - 1}{K} \le p^{*} = q^{*} \le \frac{k^{*}}{K}$. For the second equation, we deduce that
\begin{align*}
    \textstyle
    \GFT^{t}(p^{*}, q^{*})
    &\textstyle ~=~ (\BVal{t} - \SVal{t}) \cdot {\bb 1}[\SVal{t} \le p^{*} \land q^{*} \le \BVal{t}] \\
    \mr{$p^{*} = q^{*}$}
    &\textstyle ~=~ [\BVal{t} - q^{*}]_{+} \cdot {\bb 1}[\SVal{t} \le p^{*}]
    + [p^{*} - \SVal{t}]_{+} \cdot {\bb 1}[q^{*} \le \BVal{t}] \\
    \mr{$\frac{k^{*} - 1}{K} \le p^{*} = q^{*} \le \frac{k^{*}}{K}$}
    &\textstyle ~\le~ [\BVal{t} - \frac{k^{*} - 1}{K}]_{+} \cdot {\bb 1}[\SVal{t} \le \frac{k^{*}}{K}]
    + [\frac{k^{*}}{K} - \SVal{t}]_{+} \cdot {\bb 1}[\frac{k^{*} - 1}{K} \le \BVal{t}]
    \hspace{.27cm} \\
    \mr{definition of $\Tilde{\GFT}_{k^{*}}^{t}$}
    &\textstyle ~=~ \Tilde{\GFT}_{k^{*}}^{t}.
    \qedhere
\end{align*}
\end{proof}

\begin{algorithm}[t]
\caption{\label{alg:GBB-adversarial}
{\GBBSemi}}
\begin{algorithmic}[1]
    \State Run the subroutine $\ProfitMax(K, \beta)$ for $T' \in [T]$ rounds.
    \Comment{Cf.\ \Cref{prop:BCCF24,cor:BCCF24:GBB-adversarial}.}
    \label{alg:GBB-adversarial:profit}
    
    \For{$t \in [T' + 1 : T]$}
    \label{alg:GBB-adversarial:loop}
        \State $w_{k}^{t} \gets \frac{\exp(\eta \cdot \sum_{r \in [T' + 1 : t - 1]} \Hat{\GFT}_{k}^{r})}{\sum_{i \in [K]} \exp(\eta \cdot \sum_{r \in [T' + 1 : t - 1]} \Hat{\GFT}_{i}^{r})}$, $\forall k \in [K]$.
        \Comment{The distribution $w^{t} = (w_{k}^{t})_{k \in [K]}$ for Line~\ref{alg:GBB-adversarial:exploitation}.}
        \label{alg:GBB-adversarial:distribution}
        
        \State $A^{t} \sim \Bern(\gamma)$.
        \label{alg:GBB-adversarial:event}
        
        \If{$A^{t} = 1$}
        \Comment{A ``right-boundary'' round with $\Price{t} \in \{1\} \times [0, 1]$.}
        \label{alg:GBB-adversarial:case1}
            \State $k^{t} \gets \perp$, take action $(\SPrice{t} \gets 1, \BPrice{t} \sim \Unif[0, 1])$, thus semi feedback $(\SVal{t}, \Trade[]{t})$.
            \label{alg:GBB-adversarial:exploration}
        \Else
        \Comment{A ``near-diagonal'' round with $\Price{t} \in \{(\frac{k}{K}, \frac{k - 1}{K})\}_{k \in [K]}$.}
        \label{alg:GBB-adversarial:case0}
            \State $k^{t} \sim w^{t}$, take action $(\SPrice{t} \gets \frac{k^{t}}{K}, \BPrice{t} \gets \frac{k^{t} - 1}{K})$, thus semi feedback $(\SVal{t}, \Trade[]{t})$.
            \label{alg:GBB-adversarial:exploitation}
        \EndIf
        
        \State $\Hat{\GFT}_{k}^{t} \gets \big(1 - \frac{A^{t}}{\gamma} \cdot (1 - {\bb 1}[\SVal{t} \le \frac{k}{K} \land \frac{k - 1}{K} \le \BPrice{t}] \cdot \Trade[]{t})\big)
        + \big(1 - \frac{1 - A^{t}}{1 - \gamma} \cdot \frac{{\bb 1}[k^{t} = k]}{w_{k}^{t}} \cdot (1 - [\frac{k}{K} - \SVal{t}]_{+} \cdot \Trade[]{t})\big)$, $\forall k \in [K]$.
        \label{alg:GBB-adversarial:estimation}
    \EndFor
\end{algorithmic}
\end{algorithm}

\subsection{Performance Analysis of {\GBBSemi}}

To reason about the first phase (Line~\ref{alg:GBB-adversarial:profit}) of {\GBBSemi}, we only need the following \Cref{prop:BCCF24} about the subroutine {\ProfitMax}, which is quoted (or, indeed, slightly rephrased) from \cite[Lemma~5.1]{BCCF24}.
For more details, the interested reader can refer to \cite[Sections 3 and 5]{BCCF24}.

\begin{proposition}[{\cite[Lemma~5.1]{BCCF24}}]
\label{prop:BCCF24}
\begin{flushleft}
There exists a fixed-price mechanism $\ProfitMax(K', \beta')$ with one-bit feedback, on input a discretization parameter $K' \ge 1$ and a profit threshold $\beta' > 0$, such that:
\begin{enumerate}
    \item\label{prop:BCCF24:1}
    It takes actions $\{\Price{t}\}_{t = 1, 2, \dots}$ only from a size-$\abs{\+{F}_{K'}} = 2K'(\log(T) + 1)$ discrete subset $\+{F}_{K'} \subseteq \{(p, q) \in [0, 1]^{2} \;|\; p \le q\}$ of the upper-left action halfspace.\\
    Thus, the per-round profit is always nonnegative $\Profit^{t}\Price{t} \ge 0$, $\forall t = 1, 2, \dots$.
    
    \item\label{prop:BCCF24:2}
    It terminates at the end of some round $T' \in [T]$, which has two possibilities:\\
    (i)~$T' \in [T]$ is the first round such that $\sum_{t \in [T']} \Profit^{t}\Price{t} \ge \beta'$, if existential.\\
    (ii)~$T' = T$, if $\sum_{t \in [T]} \Profit^{t}\Price{t} < \beta'$.\\
    With probability $1 - T^{-1}$, each case has cumulative regret
    \begin{align*}
        \textstyle
        \sum_{t \in [T']} \Regret^{t}\Price{t}
        ~\le~ (8\beta' + 8)\log(T) + \frac{5T}{K'} + 256\sqrt{T \abs{\+{F}_{K'}} \log(T \abs{\+{F}_{K'}})} \cdot \log(T).
    \end{align*}
\end{enumerate}
\end{flushleft}
\end{proposition}

\noindent
Based on \Cref{prop:BCCF24}, we can infer the following \Cref{cor:BCCF24:GBB-adversarial} via elementary algebra; we omit its formal proof for brevity.
Since this regret bound holds with probability $1 - T^{-1}$, we assume so hereafter.

\begin{corollary}[{\ProfitMax}; Instantiation]
\label{cor:BCCF24:GBB-adversarial}
\begin{flushleft}
In the context of \Cref{prop:BCCF24}, set $K' \gets K$ and $\beta' \gets \beta$.
With probability $1 - T^{-1}$, each case has cumulative regret $\sum_{t \in [T']} \Regret^{t}\Price{t} \le 306T^{2 / 3}\log^{5 / 3}(T)$.
\end{flushleft}
\end{corollary}

Below, \Cref{lem:GFT-adversarial:GBB,rmk:GFT-adversarial:GBB} verify that {\GBBSemi} satisfies the {\GBB} constraint.

\begin{lemma}[{\GBBSemi}; The {\GBB} Constraint]
\label{lem:GFT-adversarial:GBB}
\begin{flushleft}
With probability $1 - T^{-1}$, throughout the whole fixed-price mechanism {\GBBSemi}, the total profit is nonnegative $\sum_{t \in [T]} \Profit^{t}\Price{t} \ge 0$.
\end{flushleft}
\end{lemma}

\begin{proof}
The first phase runs for $T' \in [T]$ rounds and terminates in one of two cases (\Cref{prop:BCCF24:2} of \Cref{prop:BCCF24}):\\
(i)~$T' \in [T]$ is the first round such that $\sum_{t \in [T']} \Profit^{t}\Price{t} \ge \beta$, if existential.\\
In this case, it suffices to show $\Pr{\sum_{t \in [T' + 1 : T]} \Profit^{t}\Price{t} < -\beta} \le T^{-1}$ for the second phase. Note that:\\
(Lines~\ref{alg:GBB-adversarial:case1} and \ref{alg:GBB-adversarial:exploration}) A right-boundary round $t \in [T' + 1 : T]$ with $A^{t} = 1$ takes action $(\SPrice{t} = 1, \BPrice{t} \sim \Unif[0, 1])$, so the per-round profit $\Profit^{t}\Price{t} \ge -|\SPrice{t} - \BPrice{t}| \ge -1$.\\
(Lines~\ref{alg:GBB-adversarial:case0} and \ref{alg:GBB-adversarial:exploitation}) A near-diagonal round $t \in [T' + 1 : T]$ with $A^{t} = 0$ takes action $(\SPrice{t} = \frac{k^{t}}{K}, \BPrice{t} = \frac{k^{t} - 1}{K})$, so the per-round profit $\Profit^{t}\Price{t} \ge -|\SPrice{t} - \BPrice{t}| = -\frac{1}{K}$.\\
Consequently, we can upper-bound the considered probability as follows:
\begin{align*}
    \textstyle
    \Pr{\sum_{t \in [T' + 1 : T]} \Profit^{t}\Price{t} < -\beta}
    &\textstyle ~\le~ \Pr{\sum_{t \in [T' + 1 : T]} \big(A^{t} \cdot (-1) + (1 - A^{t}) \cdot (-\frac{1}{K})\big) < -\beta} \\
    \mr{$K = \frac{3T}{\beta} - 1$}
    &\textstyle ~\le~ \Pr{\sum_{t \in [T' + 1 : T]} A^{t} > \frac{2}{3}\beta} \\
    \mr{Line~\ref{alg:GBB-adversarial:event}, $\gamma = \frac{\beta}{3T}$, and Chernoff bound}
    &\textstyle ~\le~ \exp(-\frac{1}{9}\beta) \\
    \mr{$\beta = \tTheta(T^{2 / 3})$ and $T \gg 1$}
    &\textstyle ~\le~ T^{-1}.
\end{align*}

\noindent
(ii)~$T' = T$, if $\sum_{t \in [T]} \Profit^{t}\Price{t} < \beta$.
\Comment{I.e., {\GBBSemi} even does not enter the second phase.}\\
In this case, we have $\Profit^{t}\Price{t} \ge 0$, $\forall t \in [T]$ (\Cref{prop:BCCF24:1} of \Cref{prop:BCCF24}), so the {\GBB} constraint holds.

This finishes the proof of \Cref{lem:GFT-adversarial:GBB}, covering both cases.
\end{proof}

\begin{remark}[{\GBBSemi}; The {\GBB} Constraint]
\label{rmk:GFT-adversarial:GBB}
A slight modification of our fixed-price mechanism {\GBBSemi} will guarantee the {\GBB} constraint.
Namely, in the second phase (Lines~\ref{alg:GBB-adversarial:loop} to \ref{alg:GBB-adversarial:estimation}), once the cumulative profit $\sum_{t \in [T'']} \Profit^{t}\Price{t}$ drops (from $\ge \beta = \tTheta(T^{2 / 3})$) to $\le 1$ after some round $T'' \in [T' + 1 : T]$,\footnote{At this moment, we still have $\sum_{t \in [T'']} \Profit^{t}\Price{t} \ge 0$, since the per-round profit is bounded between $[-1, 1]$.}
we instead take any ``diagonal'' actions $\SPrice{t} = \BPrice{t}$ in all remaining rounds $t \in [T'' + 1 : T]$.
Reusing the arguments in the proof \Cref{lem:GFT-adversarial:GBB}, this modification only incurs $\le T^{-1} \cdot T = 1$ additional regret.
\end{remark}

The rest of \Cref{sec:GBB-adversarial} is devoted to establishing an $\tO(T^{2 / 3})$ regret bound for {\GBBSemi}.
First of all, the following \Cref{lem:GFT-adversarial:estimation} shows that the estimates $\Hat{\GFT}_{k}^{t}$ derived in Line~\ref{alg:GBB-adversarial:estimation} are \textit{unbiased estimators} for $\Tilde{\GFT}_{k}^{t}$.

\begin{lemma}[{\GBBSemi}; Estimates in Line~\ref{alg:GBB-adversarial:estimation}]
\label{lem:GFT-adversarial:estimation}
\begin{flushleft}
For each round $t \in [T' + 1 : T]$ and conditioned on any realization of the fixed-price mechanism {\GBBSemi} in the earlier rounds $r \in [t - 1]$:
\begin{align*}
    & {\bb E}_{\Price{t}}\big[\Hat{\GFT}_{k}^{t}\big] ~=~ \Tilde{\GFT}_{k}^{t},
    && \forall k \in [K].
\end{align*}
\end{flushleft}
\end{lemma}

\begin{proof}
We deduce from Line~\ref{alg:GBB-adversarial:estimation} that
\begin{align*}
    \textstyle
    {\bb E}_{\Price{t}}\big[\Hat{\GFT}_{k}^{t}\big]
    &\textstyle ~=~ {\bb E}_{\Price{t}}\big[\big(1 - \frac{A^{t}}{\gamma} \cdot (1 - {\bb 1}[\SVal{t} \le \frac{k}{K} \land \frac{k - 1}{K} \le \BPrice{t}] \cdot \Trade[]{t})\big) \\
    &\textstyle \phantom{~=~}\qquad + \big(1 - \frac{1 - A^{t}}{1 - \gamma} \cdot \frac{{\bb 1}[k^{t} = k]}{w_{k}^{t}} \cdot (1 - [\frac{k}{K} - \SVal{t}]_{+} \cdot \Trade[]{t})\big)\big] \\
    \mr{linearity of expectation}
    &\textstyle ~=~ 2 - {\bb E}_{\Price{t}}\big[\frac{A^{t}}{\gamma} \cdot (1 - {\bb 1}[\SVal{t} \le \frac{k}{K} \land \frac{k - 1}{K} \le \BPrice{t}] \cdot \Trade[]{t})\big] \\
    &\textstyle \phantom{~=~}\qquad - {\bb E}_{\Price{t}}\big[\frac{1 - A^{t}}{1 - \gamma} \cdot \frac{{\bb 1}[k^{t} = k]}{w_{k}^{t}} \cdot (1 - [\frac{k}{K} - \SVal{t}]_{+} \cdot \Trade[]{t})\big] \\
    \mr{Lines~\ref{alg:GBB-adversarial:event} and \ref{alg:GBB-adversarial:exploitation}}
    &\textstyle ~=~ {\bb E}_{\Price{t}}\big[{\bb 1}[\SVal{t} \le \frac{k}{K} \land \frac{k - 1}{K} \le \BPrice{t}] \cdot \Trade[]{t} \;\big|\; A^{t} = 1\big] \\
    &\textstyle \phantom{~=~}\qquad + {\bb E}_{\Price{t}}\big[[\frac{k}{K} - \SVal{t}]_{+} \cdot \Trade[]{t} \;\big|\; A^{t} = 0 \land k^{t} = k\big] \\
    \mr{Lines~\ref{alg:GBB-adversarial:exploration} and \ref{alg:GBB-adversarial:exploitation}}
    &\textstyle ~=~ {\bb E}_{\BPrice{t} \sim \Unif[0, 1]}\big[{\bb 1}[\SVal{t} \le \frac{k}{K} \land \frac{k - 1}{K} \le \BPrice{t}] \cdot {\bb 1}[\SVal{t} \le 1 \land \BPrice{t} \le \BVal{t}]\big] \\
    &\textstyle \phantom{~=~}\qquad + [\frac{k}{K} - \SVal{t}]_{+} \cdot {\bb 1}[\SVal{t} \le \frac{k}{K} \land \frac{k - 1}{K} \le \BVal{t}] \\
    \mr{$\Val{t} \in [0, 1]^{2}$}
    &\textstyle ~=~ [\BVal{t} - \frac{k - 1}{K}]_{+} \cdot {\bb 1}[\SVal{t} \le \frac{k}{K}]
    + [\frac{k}{K} - \SVal{t}]_{+} \cdot {\bb 1}[\frac{k - 1}{K} \le \BVal{t}] \\
    \mr{definition of $\Tilde{\GFT}_{k}^{t}$}
    &\textstyle ~=~ \Tilde{\GFT}_{k}^{t}.
    \qedhere
\end{align*}
\end{proof}

\Cref{lem:GFT-adversarial:estimation} suggests that $\Hat{\GFT}_{k^{*}}^{t}$ is a good estimator for the benchmark $\Tilde{\GFT}_{k^{*}}^{t} \ge \GFT(p^{*}, q^{*})$ (\Cref{lem:GFT-adversarial:discretization}), and that $\langle w^{t}, \Hat{\GFT}^{t} \rangle = \sum_{k \in [K]} w_{k}^{t} \cdot \Hat{\GFT}_{k}^{t}$ is a good estimator for the per-round {\GainsFromTrade} (especially in a near-diagonal round $t \in [T' + 1 : T]$ with (Line~\ref{alg:GBB-adversarial:exploitation}) $A^{t} = 0$, $k^{t} \sim w^{t}$, and $\Price{t} = (\frac{k^{t}}{K}, \frac{k^{t} - 1}{K})$).
Hence, the quantity $\sum_{t \in [T' + 1 : T]} (\Hat{\GFT}_{k^{*}}^{t} - \langle w^{t}, \Hat{\GFT}^{t} \rangle)$ serves as a good estimator for the regret by the second phase (Lines~\ref{alg:GBB-adversarial:loop} to \ref{alg:GBB-adversarial:estimation}); the following \Cref{lem:GBB-adversarial:exploitation} establishes a useful upper bound for this quantity.

\begin{lemma}[{\GBBSemi}; Estimate for the Regret by the Second Phase]
\label{lem:GBB-adversarial:exploitation}
\begin{flushleft}
For any realization of the fixed-price mechanism {\GBBSemi} throughout the whole time horizon $t \in [T]$:
\begin{align*}
    \textstyle
    \sum_{t \in [T' + 1 : T]} \big(\Hat{\GFT}_{k^{*}}^{t} - \langle w^{t}, \Hat{\GFT}^{t} \rangle\big)
    ~\le~ \frac{\ln(K)}{\eta} + \frac{\eta}{2} \cdot \sum_{t \in [T' + 1 : T]} \sum_{k \in [K]} w_{k}^{t} \cdot (2 - \Hat{\GFT}_{k}^{t})^{2}.
\end{align*}
\end{flushleft}
\end{lemma}

\begin{proof}
For ease of notation, we denote $W^{t} \defeq \sum_{k \in [K]} \exp(\eta \cdot \sum_{r \in [T' + 1 : t - 1]} \Hat{\GFT}_{k}^{r})$, $\forall t \in [T' + 1 : T + 1]$.

On the one hand, we can lower-bound $\ln(\frac{W^{T + 1}}{W^{T' + 1}})$ as follows:
\begin{align*}
    \textstyle
    \ln(\frac{W^{T + 1}}{W^{T' + 1}})
    &\textstyle ~=~ \ln(\sum_{k \in [K]} \exp(\eta \cdot \sum_{r \in [T' + 1 : T]} \Hat{\GFT}_{k}^{r})) - \ln(K)
    \hspace{4.62cm} \\
    &\textstyle ~\ge~ \eta \cdot \sum_{r \in [T' + 1 : T]} \Hat{\GFT}_{k^{*}}^{r} - \ln(K).
\end{align*}

On the other hand, we can upper-bound $\ln(\frac{W^{T + 1}}{W^{T' + 1}})$ as follows:
\begin{align*}
    \textstyle
    \ln(\frac{W^{T + 1}}{W^{T' + 1}})
    &\textstyle ~=~ \sum_{t \in [T' + 1 : T]} \ln(\frac{W^{t + 1}}{W^{t}}) \\
    &\textstyle ~=~ \sum_{t \in [T' + 1 : T]} \ln(\sum_{k \in [K]} w_{k}^{t} \cdot \exp(\eta \cdot \Hat{\GFT}_{k}^{t})) \\
    &\textstyle ~=~ \sum_{t \in [T' + 1 : T]} \Big(2\eta + \ln(\sum_{k \in [K]} w_{k}^{t} \cdot \exp(\eta \cdot (\Hat{\GFT}_{k}^{t} - 2)))\Big) \\
    &\textstyle ~\le~ \sum_{t \in [T' + 1 : T]} \Big(2\eta + \ln(\sum_{k \in [K]} w_{k}^{t} \cdot \big(1 + \eta \cdot (\Hat{\GFT}_{k}^{t} - 2) + \frac{\eta^{2}}{2} \cdot (\Hat{\GFT}_{k}^{t} - 2)^{2}\big))\Big) \\
    &\textstyle ~=~ \sum_{t \in [T' + 1 : T]} \Big(2\eta + \ln(1 - 2\eta + \eta \cdot \langle w^{t}, \Hat{\GFT}^{t} \rangle + \frac{\eta^{2}}{2} \cdot \sum_{k \in [K]} w_{k}^{t} \cdot (\Hat{\GFT}_{k}^{t} - 2)^{2})\Big) \\
    &\textstyle ~\le~ \eta \cdot \sum_{t \in [T' + 1 : T]} \langle w^{t}, \Hat{\GFT}^{t} \rangle + \frac{\eta^{2}}{2} \cdot \sum_{t \in [T' + 1 : T]} \sum_{k \in [K]} w_{k}^{t} \cdot (\Hat{\GFT}_{k}^{t} - 2)^{2}.
\end{align*}
Here the second step applies the defining formulae of $W^{t}$, $W^{t + 1}$, and $w_{k}^{t}$ (Line~\ref{alg:GBB-adversarial:distribution}).
The fourth step applies $\Hat{\GFT}_{k}^{t} - 2 \le 0$ (Lines~\ref{alg:GBB-adversarial:estimation}) and $e^{x} \le 1 + x + \frac{1}{2}x^{2}$, $\forall x \le 0$.
And the last step applies $\ln(1 + x) \le x$.

Combining and rearranging the two equations finishes the proof of \Cref{lem:GBB-adversarial:exploitation}.
\end{proof}

Next, \Cref{lem:GBB-adversarial:second-moment} evaluates the expectations of the summands in the above upper-bound formula.

\begin{lemma}[{\GBBSemi}; Second Moment Bounds]
\label{lem:GBB-adversarial:second-moment}
\begin{flushleft}
For each round $t \in [T' + 1 : T]$ and conditioned on any realization of the fixed-price mechanism {\GBBSemi} in the earlier rounds $r \in [t - 1]$:
\begin{align*}
    \textstyle
    {\bb E}_{\Price{t}}\big[\sum_{k \in [K]} w_{k}^{t} \cdot (2 - \Hat{\GFT}_{k}^{t})^{2}\big] ~\le~ 2K + 2.
\end{align*}
\end{flushleft}
\end{lemma}

\begin{proof}
In a right-boundary round $t \in [T' + 1 : T]$ with $A^{t} = 1$, we deduce from Line~\ref{alg:GBB-adversarial:estimation} that
\begin{align*}
    \textstyle
    \sum_{k \in [K]} w_{k}^{t} \cdot (2 - \Hat{\GFT}_{k}^{t})^{2}
    &\textstyle ~=~ \sum_{k \in [K]} w_{k}^{t} \cdot \frac{1}{\gamma^{2}} \cdot (1 - {\bb 1}[\SVal{t} \le \frac{k}{K} \land \frac{k - 1}{K} \le \BPrice{t}] \cdot \Trade[]{t})^{2}
    \hspace{2.04cm} \\
    \mr{${\bb 1}[\SVal{t} \le \frac{k}{K} \land \frac{k - 1}{K} \le \BPrice{t}] \cdot \Trade[]{t} \in [0, 1]$}
    &\textstyle ~\le~ \sum_{k \in [K]} w_{k}^{t}  \cdot \frac{1}{\gamma^{2}} \\
    &\textstyle ~=~ \frac{1}{\gamma^{2}}.
\end{align*}
In a near-diagonal round $t \in [T' + 1 : T]$ with $A^{t} = 0$, we deduce from Line~\ref{alg:GBB-adversarial:estimation} that
\begin{align*}
    \textstyle
    \sum_{k \in [K]} w_{k}^{t} \cdot (2 - \Hat{\GFT}_{k}^{t})^{2}
    &\textstyle ~=~ \sum_{k \in [K]} w_{k}^{t} \cdot \frac{1}{(1 - \gamma)^{2}} \cdot (\frac{{\bb 1}[k = k^{t}]}{w_{k}^{t}} \cdot [\frac{k}{K} - \SVal{t}]_{+} \cdot \Trade[]{t})^{2}
    \hspace{1.06cm} \\
    &\textstyle ~=~ \frac{1}{(1 - \gamma)^{2} \cdot w_{k^{t}}^{t}} \cdot ([\frac{k^{t}}{K} - \SVal{t}]_{+} \cdot \Trade[]{t})^{2} \\
    \mr{$[\frac{k^{t}}{K} - \SVal{t}]_{+} \cdot \Trade[]{t} \in [0, 1]$}
    &\textstyle ~\le~ \frac{1}{(1 - \gamma)^{2} \cdot w_{k^{t}}^{t}}.
\end{align*}
In combination, we deduce from Line~\ref{alg:GBB-adversarial:distribution} that
\begin{align*}
    \textstyle
    {\bb E}_{\Price{t}}\big[\sum_{k \in [K]} w_{k}^{t} \cdot (2 - \Hat{\GFT}_{k}^{t})^{2}\big]
    &\textstyle ~\le~ \gamma \cdot {\bb E}_{\Price{t}}[\frac{1}{\gamma^{2}} \mid A^{t} = 1]
    + (1 - \gamma) \cdot {\bb E}_{\Price{t}}[\frac{1}{(1 - \gamma)^{2} \cdot w_{k^{t}}^{t}} \mid A^{t} = 0] \\
    \mr{Line~\ref{alg:GBB-adversarial:exploration}}
    &\textstyle ~=~ \gamma \cdot \frac{1}{\gamma^{2}} + (1 - \gamma) \cdot \sum_{k \in [K]} w_{k}^{t} \cdot \frac{1}{(1 - \gamma)^{2} \cdot w_{k^{t}}^{t}} \\
    \mr{$\gamma = \frac{1}{K + 1}$}
    &\textstyle ~=~ 2K + 2.
    \qedhere
\end{align*}
\end{proof}

Finally, we are ready to show the performance guarantees of our fixed-price mechanism {\GBBSemi}.

\begin{proof}[Proof of \Cref{thm:GBB-adversarial}]
Since the {\GBB} constraint has been verified in \Cref{lem:GFT-adversarial:GBB,rmk:GFT-adversarial:GBB}, we only prove an $\tO(T^{2 / 3})$ regret bound.
For the first phase (Line~\ref{alg:GBB-adversarial:profit}), we know from \Cref{cor:BCCF24:GBB-adversarial} that
\begin{align*}
    \textstyle
    \text{Regret by the first phase}
    ~\le~ 306T^{2 / 3}\log^{5 / 3}(T).
    \hspace{4.89cm}
\end{align*}
For the second phase (Lines~\ref{alg:GBB-adversarial:loop} to \ref{alg:GBB-adversarial:estimation}), by the linearity of expectation, we can deduce that
\begin{align}
    \text{Regret by the second phase}
    &\textstyle ~=~ {\bb E}\big[\sum_{t \in [T' + 1 : T]} \big(\GFT^{t}(p^{*}, q^{*}) - \GFT^{t}\Price{t}\big)\big]
    \notag \\
    &\textstyle ~=~ {\bb E}\big[\sum_{t \in [T' + 1 : T]} \big(\GFT^{t}(p^{*}, q^{*}) - \Tilde{\GFT}_{k^{*}}^{t}\big)\big]
    \label{eq:GBB-adversarial:1} \\
    &\textstyle \phantom{~=~}\quad + {\bb E}\big[\sum_{t \in [T' + 1 : T]} \big(\Tilde{\GFT}_{k^{*}}^{t} - \langle w^{t}, \Hat{\GFT}^{t} \rangle\big)\big]
    \label{eq:GBB-adversarial:2} \\
    &\textstyle \phantom{~=~}\quad + {\bb E}\big[\sum_{t \in [T' + 1 : T]} \big(\langle w^{t}, \Hat{\GFT}^{t} \rangle - \langle w^{t}, \Tilde{\GFT}^{t} \rangle\big)\big]
    \label{eq:GBB-adversarial:3} \\
    &\textstyle \phantom{~=~}\quad + {\bb E}\big[\sum_{t \in [T' + 1 : T]: A^{t} = 1} \big(\langle w^{t}, \Tilde{\GFT}^{t} \rangle - \GFT^{t}\Price{t}\big)\big]
    \label{eq:GBB-adversarial:4} \\
    &\textstyle \phantom{~=~}\quad + {\bb E}\big[\sum_{t \in [T' + 1 : T]: A^{t} = 0} \big(\langle w^{t}, \Tilde{\GFT}^{t} \rangle - \GFT^{t}\Price{t}\big)\big].
    \label{eq:GBB-adversarial:5}
\end{align}
We bound each term as follows:
\begin{align*}
    \eqref{eq:GBB-adversarial:1}
    &\textstyle ~\le~ 0.
    \tag{\Cref{lem:GFT-adversarial:discretization}} \\
    \eqref{eq:GBB-adversarial:2}
    &\textstyle ~\le~ \frac{\ln(K)}{\eta} + \frac{\eta}{2} \cdot {\bb E}\big[\sum_{t \in [T' + 1 : T]} {\bb E}_{\Price{t}}\big[\sum_{k \in [K]} w_{k}^{t} \cdot (2 - \Hat{\GFT}_{k}^{t})^{2}\big]\big]
    \tag{\Cref{lem:GBB-adversarial:exploitation}} \\
    &\textstyle ~\le~ \frac{\ln(K)}{\eta} + \frac{\eta}{2} \cdot {\bb E}\big[\sum_{t \in [T' + 1 : T]} (2K + 2)\big]
    \tag{\Cref{lem:GBB-adversarial:second-moment}} \\
    &\textstyle ~\le~ \frac{\ln(K)}{\eta} + \eta T (K + 1). \\
    \eqref{eq:GBB-adversarial:3}
    &\textstyle ~=~ {\bb E}\big[\sum_{t \in [T' + 1 : T]} \sum_{k \in [K]} w_{k}^{t} \cdot \big({\bb E}_{\Price{t}}\big[\Hat{\GFT}_{k}^{t}\big] - \Tilde{\GFT}_{k}^{t}\big)\big]
    \tag{linearity of expectation} \\
    &\textstyle ~=~ 0.
    \tag{\Cref{lem:GFT-adversarial:estimation}} \\
    \eqref{eq:GBB-adversarial:4}
    &\textstyle ~\le~ {\bb E}\big[\sum_{t \in [T' + 1 : T]: A^{t} = 1} \big((1 + \frac{1}{K}) - 0\big)\big]
    \tag{\Cref{lem:GFT-adversarial:discretization}} \\
    &\textstyle ~\le~ \gamma T \cdot (1 + \frac{1}{K}).
    \tag{Line~\ref{alg:GBB-adversarial:event}} \\
    \eqref{eq:GBB-adversarial:5}
    &\textstyle ~=~ {\bb E}\big[\sum_{t \in [T' + 1 : T]: A^{t} = 0} \sum_{k \in [K]} w_{k}^{t} \cdot \big(\Tilde{\GFT}_{k}^{t} - \GFT^{t}(\frac{k}{K}, \frac{k - 1}{K})\big)\big]
    \tag{Line~\ref{alg:GBB-adversarial:exploitation}} \\
    &\textstyle ~\le~ {\bb E}\big[\sum_{t \in [T' + 1 : T]: A^{t} = 0} \sum_{k \in [K]} w_{k}^{t} \cdot \frac{1}{K}\big]
    \tag{\Cref{lem:GFT-adversarial:discretization}} \\
    &\textstyle ~\le~ (1 - \gamma) T \cdot \frac{1}{K}.
    \tag{Line~\ref{alg:GBB-adversarial:event}}
\end{align*}
Then, by substituting $K = \frac{1}{4}T^{1 / 3} \log^{-2 / 3}(T)$, $\eta = (\frac{\ln(K)}{T (K + 1)})^{1 / 2}$, and $\gamma = \frac{1}{K + 1}$, it follows that
\begin{align*}
    \text{Regret by the second phase}
    &\textstyle ~\le~ \frac{\ln(K)}{\eta} + \eta T(K + 1)
    + \gamma T \cdot (1 + \frac{1}{K})
    + (1 - \gamma) T \cdot \frac{1}{K}
    \hspace{.57cm} \\
    &\textstyle ~=~ (8 \pm o(1)) T^{2 / 3} \log^{2 / 3}(T) \\
    \mr{$T \gg 1$}
    &\textstyle ~\le~ 9T^{2 / 3} \log^{2 / 3}(T).
\end{align*}
Combining the regret bounds for both phases finishes the proof of \Cref{thm:GBB-adversarial}.
\end{proof}

\section*{Acknowledgment}
I am grateful to Houshuang Chen and Chihao Zhang for invaluable discussions.

\bibliography{main}
\bibliographystyle{alpha}

\end{document}